\documentclass{amsart}

\usepackage[usenames,dvipsnames]{color}
\usepackage{amsmath,amsthm,amsfonts,amssymb,multicol,amscd,amsbsy}
\usepackage{graphicx}
\usepackage{booktabs}
\usepackage{array}
\usepackage{subfigure}
\usepackage[font=scriptsize]{caption}
\usepackage{enumerate}
\usepackage{url}
\usepackage[nodayofweek]{datetime}
\usepackage[english]{babel}
\usepackage{mathtools}
\usepackage[toc,page]{appendix}
\usepackage{verbatim} 
\usepackage{amsthm} 
\usepackage{enumitem}
\usepackage{enumerate}
\usepackage{bbm} 
\usepackage[normalem]{ulem} 
\usepackage{bm}

\newcommand{\be}{\begin}
\newcommand{\e}{\end}
\newcommand{\beq}{\begin{equation}}
\newcommand{\eeq}{\end{equation}}
\newcommand{\beqs}{\begin{equation*}}
\newcommand{\eeqs}{\end{equation*}}
\newcommand{\bal}{\begin{align}}
\newcommand{\eal}{\end{align}}
\newcommand{\bals}{\begin{align*}}
\newcommand{\eals}{\end{align*}}

\renewcommand{\l}{\left}
\renewcommand{\r}{\right}

\renewcommand{\d}{\mathrm{d}} 

\newcommand{\set}[1]{\mathbb{#1}}

\newcommand{\curly}[1]{\mathcal{#1}}

\newcommand{\om}{\omega}

\newcommand{\eps}{\epsilon}
\newcommand{\veps}{\varepsilon}
\newcommand{\lam}{\lambda}

\newcommand{\Gam}{\Gamma}

\newcommand{\al}{\alpha}
\newcommand{\de}{\delta}

\newcommand{\ind}{\mathbbm{1}}		

\newcommand{\ex}{\mathbb{E}}

\newcommand{\ttmatrix}[4]{\left(\be{array}{cc} #1&#2\\	#3&#4 \e{array}	\right)}
\newcommand{\tvector}[2]{\left(\be{array}{c}#1\\#2\e{array}\right)}

\newcommand{\scp}[2]{\langle#1,#2\rangle}

\newcommand{\ketbra}[2]{ |#1 \rangle \langle #2|}

\newcommand{\abs}[1]{\left\vert#1\right\vert}

\renewcommand{\it}{\infty}

\newtheorem{thm}{Theorem}[section]
\newtheorem{lm}[thm]{Lemma}

\theoremstyle{definition}

\theoremstyle{remark}

\def\dotuline{\bgroup
  \ifdim\ULdepth=\maxdimen  
   \settodepth\ULdepth{(j}\advance\ULdepth.4pt\fi
  \markoverwith{\begingroup
  \advance\ULdepth0.08ex
  \lower\ULdepth\hbox{\kern.15em .\kern.1em}%
  \endgroup}\ULon}

\def\dashuline{\bgroup
  \ifdim\ULdepth=\maxdimen  
   \settodepth\ULdepth{(j}\advance\ULdepth.4pt\fi
  \markoverwith{\kern.15em
  \vtop{\kern\ULdepth \hrule width .3em}%
  \kern.15em}\ULon}
\allowdisplaybreaks

\begin{document}

\title[Anomalous Lieb-Robinson Bounds]{On Anomalous Lieb-Robinson Bounds for the Fibonacci XY Chain}

\author[D. Damanik]{David Damanik}
\address{Mathematics Dept. MS-136, Rice University, Houston, TX 77005}
\email{damanik@rice.edu}

\author[M. Lemm]{Marius Lemm}
\address{Mathematics Dept. MC 253-37, California Institute of Technology, Pasadena, CA 91125}
\email{mlemm@caltech.edu}

\author[M. Lukic]{Milivoje Lukic}
\address{Mathematics Dept. MS-136, Rice University, Houston, TX 77005}
\email{milivoje.lukic@rice.edu}

\author[W. Yessen]{William Yessen}
\address{Mathematics Dept. MS-136, Rice University, Houston, TX 77005}
\email{yessen@rice.edu}

\thanks{D.Damanik was supported in part by NSF grants DMS--1067988 and DMS--1361625, M.Lukic was supported in part by NSF grant DMS--1301582, and W.Yessen was supported by NSF grant DMS--1304287}

\begin{abstract}
We rigorously prove a new kind of anomalous (or sub-ballistic) Lieb-Robinson bound for the isotropic XY chain with Fibonacci external magnetic field at arbitrary coupling. It is anomalous in that the usual exponential decay in $|x|-v|t|$ is replaced by exponential decay in $|x|-v|t|^\al$ with $0<\al<1$. In fact, we can characterize the values of $\al$ for which such a bound holds as those exceeding $\al_u^+$, the upper transport exponent of the one-body Fibonacci Hamiltonian. Following the approach of \cite{HSS11}, we relate Lieb-Robinson bounds to dynamical bounds for the one-body Hamiltonian corresponding to the XY chain via the Jordan-Wigner transformation; in our case the one-body Hamiltonian with Fibonacci potential. We can bound its dynamics by adapting techniques developed in  \cite{DT07, DT08, D05, DGY} to our purposes. To our knowledge, this is the first rigorous derivation of anomalous quantum many-body transport.

 Along the way, we prove a new result about the one-body Fibonacci Hamiltonian: the upper transport exponent agrees with the time-averaged upper transport exponent, see Corollary \ref{cor:ofproof}. We also explain why our method does not extend to yield anomalous Lieb-Robinson bounds of power-law type for the random dimer model.
\end{abstract}

\maketitle

\tableofcontents

\section{Introduction}

Lieb-Robinson (LR) bounds were first introduced by Lieb and Robinson in 1972 \cite{LiebRobinson72}. These bounds and their generalizations \cite{Hastings04, NachtergaeleOgataSims06, NachtergaeleRazSchleinSims09, NachtergaeleSims06} concern quantum spin systems with local Hamiltonians and establish that, approximately, quantum correlations (as expressed by commutators of local observables) propagate at most ballistically under the Heisenberg dynamics. That is, commutators of observables, which are initially supported a distance $|x|$ apart, are exponentially small in $|x|-v|t|$, where $v\geq 0$ is the so-called Lieb-Robinson velocity.  Therefore, in similarity to relativistic systems, LR bounds establish the existence of a ``light cone'' $|x|\leq v|t|$ outside of which correlations are suppressed.

Lieb-Robinson bounds greatly increased in popularity about 10 years ago, when Hastings and co-workers realized that they can be used to prove exponential clustering, a higher-dimensional analogue of the Lieb-Schultz-Mattis theorem and the famous area law for the entanglement entropy in one-dimensional quantum systems with a spectral gap \cite{HastingsKoma06, NachtergaeleSims06, Hastings04, Hastings07}. Since then, LR bounds have become an important tool in condensed-matter physics and quantum information theory, e.g.\ for understanding the structure of ground states in gapped systems.

Here, we take the perspective that LR bounds characterize the dynamics of many-body quantum systems as \emph{at most ballistic}. We consider a model, the Fibonacci XY spin chain, for which this bound on the dynamics can be improved to an anomalous (or sub-ballistic) LR bound, see Definition \ref{defn:LRalpha}, establishing anomalous many-body transport in this case. We refer to \cite{DLLY-PRL} for an exposition of our results that is geared towards a physics audience.

An older preprint version of this paper discussed the extension of our results to Sturmian models. The proof was based on methods from \cite{Marin2010}, which were
since found to be flawed \cite{DGLQ}. We believe that the extension can be proved by combining our methods with the ones in \cite{DGLQ}, but we leave this task to future work.

\section{Main Results}
\subsection{The Fibonacci XY Chain}
Given an integer $n$, we take as our Hilbert space
\beqs
\mathcal{H}_n= \bigotimes_{j=1}^n \mathfrak{h}_j.
 \eeqs
 where $\mathfrak{h}_j=\set{C}^2$ for all $j$. On $\curly{H}_n$, we consider the isotropic XY chain given by the Hamiltonian
\beqs
    H_n^{XY}= - \sum_{j=1}^{n-1} \l(\sigma^x_j \sigma^x_{j+1} +\sigma^y_j \sigma^y_{j+1}   \r) +\sum_{j=1}^n V_j \sigma^z_j
\eeqs
where $\{V_j\}$ is the \emph{Fibonacci external magnetic field} defined by
\beqs
    V_j= \lam \chi_{[1-\phi^{-1},1)} (j\phi^{-1}+\om \textnormal{ mod } 1)
\eeqs
with $\lam>0$ a coupling constant, $\om\in[0,1)$ the ``phase'' and
\beqs
    \phi= \frac{1+\sqrt{5}}{2}
\eeqs
the golden mean. The Fibonacci external field is the primary model of one-dimensional quasi-periodicity. As usual, the \emph{Pauli matrices} are
\beqs
    \sigma^x =\ttmatrix{0}{1}{1}{0},\quad \sigma^y =\ttmatrix{0}{-i}{i}{0},\quad \sigma^z =\ttmatrix{1}{0}{0}{-1},
\eeqs
and $\sigma^{x,y,z}_j$ denotes $\ind_1\otimes\ldots \ind_{j-1} \otimes \sigma^{x,y,z}\otimes \ind_{j+1}\ldots \otimes \ind_n$. For a finite set $S\subset \set{Z}_+$, we define the algebra of observables on $J$ by
\beqs
 \curly{A}_J = \bigotimes_{j\in J} \curly{B}(\mathfrak{h}_j),
\eeqs
where $\curly{B}(\mathfrak{h}_j)$ is the set of bounded linear operators on $\mathfrak{h}_j=\set{C}^2$, which is of course just the set of all complex $2\times 2$ matrices. We will often make use of the fact that for $J\subset J'$, there is a natural embedding of $\curly{A}_J$ into $\curly{A}_{J'}$ by tensoring with the identity on $J'\setminus J$. Also, we denote $\curly{A}_j\equiv \curly{A}_{\{j\}}$.

Finally, we recall that the \emph{Heisenberg dynamics} of an observable $A\in \curly{A}_J$ are defined by
\beqs
    \tau_t^n(A) = e^{it H_n^{XY}} A e^{-it H_n^{XY}}.
\eeqs

\textbf{Disclaimer }
We usually do not keep track of constants that depend on model parameters, one exception is the dependence on the parameter $\om$ as discussed later. We write $C,C',\ldots$ for constants  that may have different numerical values from line to line and $C_0,C_1,\ldots$ for constants that appear in the statement of a result.

\subsection{Anomalous Lieb-Robinson Bounds}
To phrase our results, it will be convenient to adopt the following convention for stating anomalous LR bounds (we will soon discuss what we mean by ``anomalous''):

\be{defn}
\label{defn:LRalpha}
We say that ``$LR(\al)$'' holds if there exist constants $C_0,\mu>0$ and $v\geq 0$ such that for all integers $1\leq j< j' \leq n$ and all $t>0$, we have
\beq
\label{eq:result}
\|[\tau_t^n(A),B]\|\leq C_0 \|A\| \|B\| e^{-\mu(|j'-j|-vt^{\al})}
\eeq
for all $A\in \curly{A}_j$ and all $B\in \curly{A}_{\{j',\ldots,n\}}$.
\e{defn}

\be{rmk}
\be{enumerate}[label=(\roman*)]
\item The ordinary LR bound is $LR(1)$ and by the very general analysis of \cite{NachtergaeleSims06}, $LR(1)$ holds in our case\footnote{In fact even the slightly stronger version of $LR(1)$ with $A\in\curly{A}_J$ and $C$ independent of the cardinality $|J|$ holds. }. The bound $LR(\al)$ with $\al<1$ is qualitatively stronger than $LR(1)$ (i.e.\ modulo changes in the constants $C_0,\mu$ and $v$). Indeed, it already becomes effective when $v^{1/\alpha}t<|j'-j|^{1/\al}$. Since space is discrete, we have $|j'-j|^{1/\al}\leq |j'-j|$ and so $LR(\al)$ becomes effective at smaller distances than $LR(1)$. See also the discussion following Theorem 1 in \cite{DLLY-PRL}.
\item The assumption that $A\in \curly{A}_j$ can be generalized to $A\in\curly{A}_J$ with $\max J\leq j$ using the Leibniz rule for commutators \eqref{eq:leibniz}, but at the price of having $C_0$ depend on $|J|$. For a detailed proof of this, we refer to \cite{DLY}. In a different approach, following the proof of Corollary 3.4 in \cite{HSS11}, one can derive from $LR(\al)$ a Lieb-Robinson bound which holds for all $A\in\curly{A}_J$ with $\max J\leq j$ and $C$ \emph{independent} of $|J|$, but at the price of increasing the growth in $t$ to $\int_0^t e^{\mu vs^\al} \d s$.
    \e{enumerate}
\e{rmk}

We will phrase our results in terms of the \emph{upper transport exponent} $\al_u^+$ for the Fibonacci Hamiltonian $H$. It is just one of several exponents characterizing the dynamics associated to $H$ and we will introduce these later, in Section~\ref{sect:al'}, mainly as tools.

\be{defn}
\label{defn:alu}
Let $H$ be the operator on $\ell^2(\set{Z}_+)$ defined in \eqref{eq:Hdefn} and let $\{\de_l\}_{l\geq 1}$ be the canonical basis of $\ell^2(\set{Z}_+)$. For $t>0$ and any integer $N$, we define
\beq\label{eq:PNt}
    P(N,t)=\sum_{n>N} \l| \scp{\de_1}{e^{-itH}\de_n}  \r|^2,\qquad S^+(\al)=-\limsup_{t\rightarrow \it}\frac{\log P(t^\al-1,t)}{\log t}
\eeq
as well as the \textbf{upper transport exponent}
\beq\label{eq:alphauplus}
    \al_u^+ =  \sup_{\al\geq 0}\{S^+(\al) < \infty\}
\eeq
which, roughly speaking, characterizes the propagation rate of the fastest part of the wavepacket initally localized at $\de_1$.
\e{defn}

\begin{thm}[First main result]
\label{thm:result}
Let $\lam>0$. If $\al>\al_u^+(\lam)$, then $LR(\al)$.
\end{thm}

\be{rmk}
The proof also yields the explicit formulae \eqref{eq:vdefn} for $\mu$ and for the ``Lieb-Robinson velocity'' $v$. The formula for $v$ does not yield quantitative information however, because it involves the quantity $C_\de'$, which is not determined in \cite{DGY}.
\e{rmk}

Our second main result says that the upper transport exponent is truly the ``correct'' one for the LR bound (modulo the difference between $\geq$ and $>$).

\begin{thm}[Second main result]
\label{thm:result2}
Let $\lam > 0$. If $LR(\al)$, then $\al\geq \al_u^+(\lam)$.
\end{thm}

\be{rmk}\label{r.generallowerbound}
The proof of Theorem~\ref{thm:result2} works in complete generality and has nothing to do with the Fibonacci case. That is, the implication $LR(\al) \Rightarrow \al \geq \al_u^+$ holds for general choices of the $\{ V_j \}$ in the isotropic XY chain (e.g.\ periodic $\{V_j\}$ for which $\al_u^+=1$).
\e{rmk}

\subsection{Discussion}
Let us explain why we call these Lieb-Robinson bounds ``anomalous''. The usual LR bound is $LR(1)$ and it implies that commutators are small, up to an exponential error, outside of the ``light cone'' given by $|j-j'|\leq vt$. For $v>0$, this behavior corresponds to \emph{ballistic transport} and for $v=0$ to \emph{dynamical localization}.

By contrast, Theorems \ref{thm:result} and \ref{thm:result2} say that the ``light-cone'' for the Fibonacci XY chain is changed to the set $|j-j'|\leq vt^{\al_u^+}$ with $0<\al_u^+<1$ (see Proposition \ref{prop:close} below). In other words, quantum-mechanical correlations spread \emph{sub-ballistically} for this model and whenever $\al_u^+\neq \frac{1}{2}$ such behavior is commonly referred to as \emph{anomalous transport}. To our knowledge, this is the first rigorous proof of anomalous quantum many-body transport. It is physically appealing that the upper bound on transport at the one-body level $\al_u^+$ is precisely what governs transport at the many-body level, but in light of the fact that the XY chain can be mapped to non-interacting particles, this is not so surprising (compare however with the situation for the random dimer model discussed in the final section).

Roughly speaking, the anomalous behavior of the Fibonacci XY chain is a consequence of the \emph{quasi-periodicity} of the $V_i$, which is situated in between the two extreme cases of
\be{itemize}
\item[(a)] \emph{periodic} external fields of $V_i$. These correspond to \emph{ballistic transport}, which is obvious e.g.\ for the free case $V_i\equiv0$. This case is discussed in an upcoming paper of three of the authors \cite{DLY}.
\item[(b)] \emph{disordered}, that is i.i.d.\ random, external fields of $V_i$. These were shown to lead to exponential decay of correlations \cite{KleinPerez} and zero-velocity LR bounds (i.e.\ dynamical localization) \cite{HSS11}.
\e{itemize}

To close the discussion of anomalous transport, we record some known upper and lower bounds on $\al_u^+$. For further remarks on the anomalous LR bound, we refer the interested reader to \cite{DLLY-PRL}. 

\be{prop}
\label{prop:close}
\be{enumerate}[label=(\roman*)]
\item For all $\lam>0$, we have $\al_u^+(\lam)>0$.
\item For all $\lam >\sqrt{24}$, we have
\beqs
    \al_u^+(\lam)\geq \frac{2\log \phi}{\log(2\lam+22)}
\eeqs
with $\phi=\frac{1+\sqrt{5}}{2}$.
\item For all $\lam\geq 8$, we have
\beqs
 \al_u^+\leq \frac{2\log\phi}{\log\xi(\lam)},
\eeqs
for $\lam\geq 8$ and $\xi(\lam)=\frac{1}{2}(\lam-4+\sqrt{(\lam-4)^2-12})$.
\e{enumerate}
\e{prop}

Numerically, we can use these to see, e.g., that
\beqs
0.1< \al_u^+ <0.5
\eeqs
for $12\leq \lambda\leq 7,000$. Hence, the transport exponent is truly anomalous for such $\lambda$ (we note the upper bound by $0.5$ because that particular exponent is sometimes referred to as ``diffusive'' transport and not assigned the ``anomalous'' label).

\be{proof}
For (i), see \cite{DKL00} and note that $\tilde{\al}_u^+\leq \al_u^+$ according to Lemma \ref{lm:averaging}. For (ii), see Theorem 3 of \cite{DT08} and use $\tilde{\al}_u^+\leq \al_u^+$.
For (iii), see Theorem 3 of \cite{DT07}.
\e{proof}

\subsection{Equality of Transport Exponents}
We explicitly note the following corollary of the proof of Proposition~\ref{prop:al'}, because we believe it to be of independent interest.
The averaged transport exponents $\tilde \al_u^\pm$ are defined in Section~\ref{sect:al'}.

\be{cor}
\label{cor:ofproof}
Let $\lam>0$. Then, $\al_u^+=\tilde \al_u^+=\tilde \al_u^-$
\e{cor}

The second equality was already observed in \cite{DGY}, the first one is new.

\section{The Relation to One-Body Dynamics}
\subsection{Diagonalizing the XY Chain}
We will use the standard procedure, going back to \cite{LSM61}, of diagonalizing the XY chain via the Jordan-Wigner transformation to free fermions, followed by a Bogoliubov transformation. 

We only recall what we need to establish notation for the relevant objects. For the details of the diagonalization procedure, we refer to Section~3.1 in \cite{HSS11}. The first step is to introduce the lowering operator
\beqs
 a_j = \frac{1}{2} \l(  \sigma_j^x-i \sigma^y_j  \r)
\eeqs
and its adjoint the raising operator $a_j^*$ for all $1\leq j\leq n$. The Jordan-Wigner transformation maps these to the fermion operators\footnote{This means that they satisfy the canonical anticommutation relations.}
\beq
\label{eq:JW}
    c_1 = a_1,\quad c_j = \sigma^z_1\ldots \sigma^z_{j-1} a_j\quad \text{for } 2\leq j\leq n.
\eeq
In terms of these operators, the Hamiltonian reads
\beqs
    H^{XY}_n = \sum_{j=1}^{n-1} \sum_{k=1}^n c^*_j (H_n)_{j,k} c_k
\eeqs
where we introduced the $n\times n$ matrix
\beq
\label{eq:Hndefn}
    H_n=\left(\be{array}{cccc}
    V_1 & 1 & &\\
    1 & \ddots & \ddots & \\
     & \ddots & \ddots & 1 \\
    & &1 & V_n\\
    \e{array}	\right).
\eeq
We will refer to $H_n$ as the \emph{one-body Fibonacci Hamiltonian}. To emphasize that it depends on the phase $\om\in [0,1)$ we sometimes write $H_n(\om)$.

We will heavily use that the Heisenberg dynamics of the $c_j$ operators is given in the following simple fashion.

\be{prop}[\cite{HSS11}]
\label{prop:tau}
For all $1\leq j,k\leq n$ and all $\om\in [0,1)$, we have
\beq
\label{eq:3.15}
 \tau_t^n(c_j) = \sum_{k=0}^n (e^{-2iH_n(\om)t})_{j,k} c_k.
\eeq
\e{prop}

\be{proof}
This is a consequence of formula (3.15) in \cite{HSS11} for $\tau_t^n(c_j)$, which is proved via Bogoliubov transformation.
\e{proof}

\subsection{Relating LR Bounds for the XY Chain to Fermionic LR Bounds}
The following lemma is instrumental in relating the LR bounds for the $c_j$ back to LR bounds for local observables in the XY chain. The difficulty is that the Jordan-Wigner transformation \eqref{eq:JW} is \emph{non-local}. This was overcome in \cite{HSS11} at the relatively small price of the extra sum over $k$ in \eqref{eq:ksum}. 

\be{defn}
We say that $LR_{\text{fermi}}(\al)$ holds if there exist constants $C_0,\mu>0$ and $v\geq 0$ such that for all integers $1\leq j< j'\leq n$ and all $t>0$, we have
\beqs
\|[\tau_t^n(c_j),B]\|+\|[\tau_t^n(c_j^*),B]\|\leq C_1 \|B\| e^{-\mu(|j'-j|-vt^\al)}
\eeqs
for all $B\in \curly{A}_{\{j',\ldots,n\}}$.
\e{defn}

For obvious reasons, we will make heavy use of

\be{lm}[\cite{HSS11}]
\label{lm:3.2}
$LR(\al)$ holds if and only if $LR_{\mathrm{fermi}}(\al)$ holds.
\e{lm}

We use the strategy of the proof of Theorem 3.2 in \cite{HSS11}, but we allow for a $t$-dependence of the form appearing on the right-hand side of the anomalous LR bound and we note that the argument can also be run in reverse.

\be{proof}
We first prove the ``if'' part. Let $A=a_j$. Since $(\sigma^z_i)^2=1$ and $[\sigma_i^z,\sigma^z_j]=0$ when $i\neq j$, we can easily invert the Jordan-Wigner transformation \eqref{eq:JW} to get
\beqs
 a_1 = c_1,\qquad a_j = \sigma_1^z\ldots \sigma_{j-1}^z c_j,\quad\forall j\geq 2.
\eeqs
By the automorphism property of $\tau_t^n$ and the ``Leibniz rule'' for commutators
\beq
\label{eq:leibniz}
[AB,C]= A[B,C]+ [A,C]B,
\eeq
we have
\beqs
 [\tau_t^n(a_j),B]= \tau_t^n(\sigma^z_{1})\ldots \tau_t^n(\sigma^z_{j-1}) [\tau_t^n(\sigma^z_{j}),B]
 + [\tau_t^n(\sigma^z_{1}),\ldots,\tau_t^n(\sigma^z_{j-1}),B]
 \tau_t^n(c_{j}).
\eeqs
By unitarity of $\tau_t^n$, this implies
\beqs
    \|[\tau_t^n(a_j),B]\|\leq \|[\tau_t^n(c_j),B]\| + C(j-1,B),
\eeqs
where we introduced
\beqs
    C(l,B)= \|[\tau_t^n(\sigma^z_{1}),\ldots,\tau_t^n(\sigma^z_{l}),B]\|.
\eeqs
Applying \eqref{eq:leibniz} again, we find
\beqs
    C(l,B)\leq C(l-1,B) + \|[\tau_t^n(\sigma^z_{l}),B]\|.
\eeqs
Since $\sigma^z_l=2 c_l^*c_l -\ind_{\set{C}^2}$, we get
\beqs
    C(l,B)\leq C(l-1,B) + 2 \|[\tau_t^n(c_l),B]\|+ 2 \|[\tau_t^n(c_l^*),B]\|,
\eeqs
which we can iterate to obtain
\beq
\label{eq:ksum}
    \|[\tau_t^n(a_j),B]\|\leq 2 \sum_{l=1}^j    \l( \|[\tau_t^n(c_l),B]\|+ \|[\tau_t^n(c_l^*),B]\|\r).
\eeq
By performing a geometric series, it is now obvious that $LR_\text{fermi}(\alpha)$ implies $LR(\al)$ in the special case when $A=a_j$. Extending this to all of $\curly{A}_j$, which is spanned by $\{a_j,a_j^*,a_j^*a_j,a_ja_j^*\}$, is not hard; we refer to \cite{HSS11} for the details.

The ``only if'' part follows by the exact same reasoning, since the Jordan-Wigner transform and its inverse are of the same form.
\e{proof}

\section{Proof of the Second Main Result}
We begin with the proof of the \emph{second} main result, Theorem \ref{thm:result2}.

\begin{lm}\label{thm:propagator-lower}
For any $l \in [1, n)$ and any $r\in (l, n]$, for any $t\in \mathbb{R}$,
\begin{align*}
\|[\tau_t^n(c_l),a_r^* ]\|  \geq \abs{(e^{-2iH_n(\om)t})_{l,r} }.
\end{align*}
\end{lm}

\begin{proof}
From \eqref{eq:3.15} we have
\begin{align*}
\|[\tau_t^n(c_l),a_r^* ]\|  = \left \lVert \left[\sum_{j=1}^{n} (e^{-2iH_n(\om)t})_{l,j} c_{j} ,\hspace{1mm} a_r^*\right] \right\rVert.
\end{align*}
Observe that for all $j < r$, $a_r^*$ commutes with $c_j$. Thus we can write
\[
\|[\tau_t^n(c_l),a_r^* ]\|  =  \left\lVert \left[\sum_{j \geq r}^{n} (e^{-2iH_n(\om)t})_{l,j} c_{j}, a_r^*\right] \right\rVert.
\]
Notice that for each $j > r$, $c_j = \sigma_1^{(z)}\cdots\sigma_{j-1}^{(z)}a_j$, and $a_r^*$ commutes with $a_j$ and every $\sigma_i^{(z)}$ with $i\neq r$. On the other hand, notice that for every $i$, $\sigma_i^{(z)}v=v$ with $v=\bigotimes_1^n \left(\begin{smallmatrix} 1\\ 0\end{smallmatrix}\right)$. Thus, since $a_r^*\left(\begin{smallmatrix}1\\0\end{smallmatrix}\right)=0$, we have
\begin{align*}
\left(\left[\sum_{j \geq r+1}^{n} (e^{-2iH_n(\om)t})_{l,j} c_{j}, a_r^*\right]  \right)v = 0,
\end{align*}
and
\begin{align*}
[(e^{-2iH_n(\om)t})_{l,r} c_r, a_r^*]v = (e^{-2iH_n(\om)t})_{l,r} v.
\end{align*}
Thus we have
\[
\|[\tau_t^n(c_l),a_r^* ]\|\geq \left\lVert [(e^{-2iH_n(\om)t})_{l,r} c_r, a_r^*]v \right\rVert = \left\lVert (e^{-2iH_n(\om)t})_{l,r} v\right\rVert = \left\lvert (e^{-2iH_n(\om)t})_{l,r} \right\rvert. \qedhere
\]
\end{proof}

\begin{proof}[Proof of Theorem \ref{thm:result2}]
By the previous lemma and the assumption of Theorem \ref{thm:result2}, we have
\begin{align*}
\abs{\langle \delta_l, e^{-2it H_n(\omega)}\delta_r\rangle} \leq \|[\tau^n_t(c_l), a_r^*]\|\leq Ce^{-\eta (r-l-v\abs{t}^\alpha)}.
\end{align*}
Let us shift everything so that we get
\begin{align*}
\abs{\langle \delta_l, e^{-2it H_n(\omega)}\delta_r\rangle} = \abs{\langle \delta_{r-l+1}, e^{-2it H_n(\tilde\omega)} \delta_1\rangle}
\end{align*}
(notice the change in the phase from $\omega$ to $\tilde{\omega}$ as a result of the shift, compare \eqref{eq:covariance}). Since we initially put no restrictions on $l$ and $r$, we get, for each $m$,
\begin{align}
\abs{\langle \delta_m, e^{-2it H_n(\tilde\omega)}\delta_0\rangle}^2 \leq C e^{- 2\eta (m - v\abs{t}^\alpha)}.
\end{align}
Now if we define, for all $N \leq n$,
\begin{align*}
P_\mathrm{out}^{(n)}(N, t):=\sum_{n \geq m\geq N}\abs{\langle \delta_m, e^{-it H_n(\tilde\omega)}\delta_1\rangle}^2,
\end{align*}
and set $P_\mathrm{out}^{(n)}(N, t) = 0$ for all $N > n$, we obtain
\begin{align*}
P_\mathrm{out}^{(n)}(N, 2t)\leq \sum_{n \geq m \geq N} C e^{- 2\eta (m - v\abs{t}^\alpha)} \lesssim e^{-2\eta(N - v \abs{t}^\alpha)}.
\end{align*}
By the dominated convergence theorem we get
\begin{align*}
\lim_{n \rightarrow \infty} P_\mathrm{out}^{(n)}(N, 2t) = P(N-1, 2t) \hspace{2mm} \text{for all}\hspace{2mm} t.
\end{align*}
In particular,
\begin{align*}
P(t^{\gamma}-1, t) \lesssim e^{- 2\eta (t^{\gamma} - v |t/2|^\alpha)}.
\end{align*}
Thus for all $\gamma > \alpha$, we have $P(t^{\gamma}-1, t) \lesssim e^{- \eta t^{\gamma}}$ for $t\gg 1$. By \eqref{eq:PNt} and \eqref{eq:alphauplus}, this implies $\gamma \ge \alpha_u^+$, so that $\alpha \geq \alpha_u^+$.
\end{proof}

\section{Proof of the First Main Result}

\subsection{The Transport Exponent in a New Guise}
We begin the proof of the first main result, Theorem \ref{thm:result}. It will be convenient for us to use the following quantity, which we will soon see agrees with $\al_u^+$, introduced in the recent paper \cite{DGY} to study transport exponents from a dynamical systems perspective.

We define 
\beq
\label{eq:al'defn}
    \al' \equiv \al'(\lam)= \frac{\log\l(\frac{1+\sqrt{5}}{2}\r)}{\lim_{k\rightarrow\it }\frac{1}{k}\log \min_{j=1,\ldots F_k}|x_k'(E_k^{j})|},
\eeq
where $F_k$ the $k$-th Fibonacci number, the Fibonacci trace map $x_k$ is defined in \eqref{eq:xMdefn} and $E_k^{(j)}$ is defined in \cite{DGY}. The precise definitions of these quantities are of limited relevance here, because we will use results of \cite{DGY} tailor-made for the analysis of $\al'$ as a ``black box''. The limit in \eqref{eq:al'defn} exists, by Proposition 3.7 in \cite{DGY}.

Another reason why the quantities appearing in the definition of $\al'$ are of limited relevance here, is that we will prove
\be{prop}
\label{prop:al'}
It holds that $\al'=\al_u^+$.
\e{prop}

\subsection{Key Result: Fermionic LR Bounds from One-Body Dynamics}

According to Lemma \ref{lm:3.2}, it suffices to prove LR bounds for the fermion operators $c_j$. They are established by the following key result:

\be{thm}[Key result]
\label{thm:LRc}
Let $\lam>0$. If $\al>\al'(\lam)$, then $LR_{\mathrm{fermi}}(\al')$.
\e{thm}

The first main result now follows easily:

\be{proof}[Proof of Theorem \ref{thm:result}]
This is a direct consequence of Theorem \ref{thm:LRc}, Proposition \ref{prop:al'} and Lemma \ref{lm:3.2}.
\e{proof}

It thus remains for us to prove Theorem \ref{thm:LRc} (Section~\ref{sect:key}) and Proposition~\ref{prop:al'} (Section~\ref{sect:al'}).

\section{Proof of the Key Result}
\label{sect:key}
\subsection{Strategy of Proof}
The proof of Theorem \ref{thm:LRc} is based on two main ingredients:
\be{itemize}
\item[(a)] Proposition \ref{prop:tau}, which gives a simple expression for the Heisenberg dynamics of $c_j$ in terms of the one-body Fibonacci Hamiltonian $H_n$. We recall that the reason why this useful formula holds is that the fermions described by the $c_j$ are \emph{non-interacting}.
\item[(b)]
    \emph{Dynamical upper bounds} for $H_n$, established by Damanik - Tcheremchantsev \cite{DT07, DT08, D05} and Damanik - Gorodetski - Yessen \cite{DGY}. Some of their tools, like the Dunford functional calculus approach, work for rather general one-dimensional quantum systems. However, the crucial \emph{exponential} lower bound on transfer matrix norms, which is a result of \cite{DGY} quoted here as Proposition \ref{prop:trace} is special to the Fibonacci case. Since the methods of \cite{DGY} apply to arbitrary coupling strength $\lam>0$, so do our results.
\e{itemize}

Applying (a) is trivial. Regarding (b), we need to modify the existing arguments somewhat. The \textbf{main difficulty} for us is that bounds on transport exponents involve \emph{probabilities} (see the definition of $\al_u^+$), which according to quantum theory are given as appropriate $\ell^2$-norms. In proving LR bounds however, we are naturally led to consider $\ell^1$-norms instead. Since this means we do not have orthogonality at our disposal, we need to develop an alternative approach\footnote{We make remark that using the Cauchy-Schwarz inequality to go from $\ell^1$-norms to $\ell^2$-norms, one picks up a factor $\sqrt{n}$ and hence loses the required uniformity in $n$.} and we find that combining a \emph{resolution of the identity} with \emph{Combes-Thomas estimates} works. Though we do not state them explicitly, our approach yields pointwise (i.e.\ non-summed) bounds on the matrix elements of the resolvent and therefore also of the propagator. (This is in contrast to \cite{DT07}, which does not produce pointwise bounds for the resolvent because Lemma 1 in \cite{DT07} involves the full $\|\cdot\|_2$-norm.)

 Some more minor obstructions are:

\be{itemize}

\item[(i)] Our $H_n$ is initially only defined on $\curly{H}_n$, whereas one usually works with full-line or half-line operators. Thus, we first need to extend $H_n$ to an appropriate half-line operator $H$ on $\ell^2(\set{Z}_+)$.

\item[(ii)] To analyze transport exponents, one is only interested in large-time behavior and hence assumes $t\geq 1$ throughout for technical reasons, while we need results for all $t>0$.

\item[(iii)] To simplify our analysis, we reduce to the initial state $\de_1=(1,0,\ldots)$ with an arbitrary phase $\om\in [0,1)$ via the \emph{covariance under shifts} of the half-line Fibonacci Hamiltonian, see \eqref{eq:covariance}. While it is often assumed that $\om=0$, it is well known, see \cite{DT07}, \cite{DGY}, that one can extend to arbitrary phases $\om\in [0,1)$ via the methods of \cite{D05} and this is precisely what we do.

\e{itemize}

\subsection{The Dunford Functional Calculus Approach to Quantum Dynamics}
It suffices to bound $[\tau_t^n(c_j),B]$, since one obtains the same bound for $[\tau_t^n(c_j^*),B]$ by taking adjoints. By \eqref{eq:3.15} and the fact that $[c_k,B]=0$ when $k<j'$, we get
\beqs
    [\tau_t^n(c_j),B] = \sum_{k=j'}^n \l( \l(e^{-2iH_n(\om)t}\r)_{j,k} [c_k,B] \r),
\eeqs
which implies
\beq
\label{eq:pf1}
    \|[\tau_t^n(c_j),B]\|\leq 2\|B\| \sum_{k=j'}^n \l|
    \l(e^{-2iH_n(\om)t}\r)_{j,k}\r|.
\eeq
From now on, we will work on the half-line Hilbert space
\beqs
    \curly{H}=\ell^2(\set{Z}_+),
\eeqs
to which we trivially extend $H_n$ by setting $H_n\de_l=0$ for $l>n$.
In the canonical basis $\{\de_l\}_{l\geq 1}$ of $\curly{H}$, we
can write the matrix elements of the time-evolution operator as
\beqs
     \scp{\de_j}{e^{-2iH_n(\om)t}\de_k}= \l(e^{-2iH_n(\om)t}\r)_{j,k}.
\eeqs
We follow the approach of \cite{DT08} and establish dynamical upper bounds \emph{without time-averaging} via the \emph{Dunford functional calculus} \cite{DunfordSchwartz}: For all $1\leq j,k\leq n$, we have
\beq
\label{eq:dunford}
    \scp{\de_j}{e^{-2iH_nt}\de_{k}} =
     -\frac{1}{2\pi i} \int_\Gam e^{-it z} \scp{\de_j}{\frac{1}{-2H_n-z}\de_{k}} \, \d z,
\eeq
where $z$ stands for $z\ind_{\set{Z}_+}$ and $\Gam$ is any positively oriented contour in $\set{C}$ that encloses the spectrum $\sigma(-2H_n)$. Note the slightly unconventional appearance of $-2H_n$ instead of $H_n$ on the right-hand side. We will choose the same rectangular $\Gam$ as was used
in the proof of Lemma 2 in \cite{DT08}. We choose
\beq
\label{eq:Kdefn}
    K= \min\{4,2\lam+5\}
\eeq
and observe
\beqs
\sigma(-2H_n)\subset [-K+1,K-1].
\eeqs
We recall the well-known

\be{prop}[Combes-Thomas estimate \cite{CombesThomas}]
There is a universal constant $C_T$ such that for all $1 \leq l,m \leq n$,
\beq
\label{eq:CT}
\l| \scp{\de_l}{\frac{1}{-2H_n-z} \, \de_m}\r| \leq 2 d^{-1} e^{-C_T d |l-m|},
\eeq
where
\beqs
d=\min\{\textnormal{dist}(z,\sigma(-2H_n)),1\}.
\eeqs
\e{prop}

\be{proof}
See the appendix in \cite{GKT04} for a few-line proof, which also directly extends to the case considered here where the discrete Laplacian is restricted to a box (this extension was already explicitly observed in \cite{KruegerPHD} for real $z$).
\e{proof}

When $t\geq 1$, we follow the proof of Lemma 2 in \cite{DT08} word-for-word until the last step, where we do \emph{not} use the Cauchy-Schwarz inequality. When $t<1$, we replace $t^{-1}$ by $1$ everywhere in the proof, in particular in the definition of the contour $\Gam$. Note that $|e^{-itz}|\leq e$ is still uniformly bounded for all $z\in\Gam$. The upshot is that \eqref{eq:dunford} yields the estimate
\beq
\label{eq:upshot}
    \l|\scp{\de_j}{e^{-2iH_nt}\de_{k}}\r|
    \leq C e^{-C_T (k-j)} +  C' \int_{-K}^K \l|\scp{\de_j}{\frac{1}{-2H_n-E-i\veps}\de_{k}}\r| \d E,
\eeq
for all $1\leq j< j'\leq k\leq n$. Here, we introduced the quantity
\beq
\label{eq:epsdefn}
    \veps=\min\{t^{-1},1\},
\eeq
which of course satisfies $\veps\leq 1$ for all $t>0$. Using \eqref{eq:upshot} on \eqref{eq:pf1} and performing a geometric series in the first term, we obtain
\beq
\label{eq:pf2}
    \|[\tau_t^n(c_j),B]\|
    \leq C  e^{-C_T(j'-j)}+  C' \int_{-K}^K \sum_{k=j'}^n \l|\scp{\de_j}{\frac{1}{-2H_n-E-i\veps}\de_{k}}\r| \d E.
\eeq
Clearly, we can ignore the first term in the following.

\subsection{Extension to the Half-Line Fibonacci Hamiltonian}
By expressing the time evolution in terms of resolvents, we can start modifying $H_n$ via the resolvent identity. The effect of these changes will be controlled by combining a \emph{resolution of the identity} and \emph{Combes-Thomas estimates}.

We define the half-line operator $H$ on $\ell^2(\set{Z}_+)$ as the tri-diagonal half-infinite matrix
\beq
\label{eq:Hdefn}
H=\left(\be{array}{cccc}
    V_1&1&&\\
    1&V_2& \ddots& \\
    &\ddots&\ddots&
    \e{array}	\right).
\eeq
We denote $R(z)=(-2H-z)^{-1}$ and $R_n(z)=(-2H_n-z)^{-1}$ with $z=E+i\veps$.
We recall the resolvent identity,
\beqs
R_n(z) = R(z)+R(z) 2(H_n-H) R_n(z).
\eeqs
Introducing a resolution of the identity $\sum_{l=1}^\it \ketbra{\de_l}{\de_l}$, we can bound the sum in \eqref{eq:pf2} by
\beq
\label{eq:pf3}
   \sum_{k=j'}^n \l|\scp{\de_j}{R_n(z)\de_k}\r|
   \leq \sum_{k=j'}^n \l|\scp{\de_j}{R(z)\de_k}\r|
    + \sum_{l=1}^\it \l|\scp{\de_j}{R(z)\de_l}\r| \sum_{k=j'}^n \l| \scp{\de_l}{2(H_n-H)R_n(z)\de_k}\r|.
\eeq
An important observation is that, on the one hand
\beqs
    H-H_n=\chi_n H\chi_n+\ketbra{\de_n}{\de_{n+1}}+\ketbra{\de_{n+1}}{\de_n},
\eeqs
where we used Dirac notation and wrote $\chi_n$ for the indicator function of $\set{Z}_+\setminus\{1,\ldots,n\}$. On the other hand,
for $j'\leq k\leq n$, $R_n(z) \de_k$ is supported in $\curly{H}_n$ due to the block diagonal structure of $H_n-z\ind_{\set{Z}_+}$. Together, these imply that the only contribution to the $l$-sum in \eqref{eq:pf3} comes from the $l=n+1$ term. Hence,
\beq
\label{eq:pf4}
    \sum_{k=j'}^n \l|\scp{\de_j}{R_n(z)\de_k}\r|
   \leq \sum_{k=j'}^n \l|\scp{\de_j}{R(z)\de_k}\r| + 2 \l|\scp{\de_j}{R(z)\de_{n+1}}\r|
   \sum_{k=j'}^n \l| \scp{\de_n}{R_n(z)\de_k}\r|.
\eeq
We apply the Combes-Thomas estimate \eqref{eq:CT} to \eqref{eq:pf4} and get
\beq
\label{eq:pf5}
\begin{aligned}
    \sum_{k=j'}^n \l|\scp{\de_j}{R_n(z)\de_k}\r|
   &\leq 2 \sum_{k=j'}^{n+1} \l|\scp{\de_j}{R(z)\de_k}\r| \l(  1+2\veps^{-1}\sum_{k=j'}^{n} e^{- C_T\veps |k-n|}  \r)\\
   &\leq C\frac{1}{\veps (1-e^{-C_T\veps})} \sum_{k=j'}^{\it} \l|\scp{\de_j}{R(z)\de_k}\r|.
   \end{aligned}
\eeq
Note that $n$ has disappeared from the last expression. Since we are aiming for uniformity in $n$, this is un-problematic.
Using \eqref{eq:pf5} to estimate \eqref{eq:pf2} and recalling $z=E+i\veps$, we see that it remains to control
\beq
\label{eq:pf6}
   C \frac{1}{\veps (1-e^{-C_1\veps})} \int_{-K}^K \sum_{k=j'}^{\it} \l|\scp{\de_j}{R(E+i\veps)\de_k}\r| \d E.
\eeq

\subsection{Covariance of $H(\om)$ under Shifts}
For $l\geq 0$, let $T_l$ be the right-shift operator defined by
\beqs
T_l\de_m=\de_{m+l}
\eeqs
with adjoint operator given by the left-shift $T_l^*=T_{-l}$.
We observe the following \emph{covariance property} of $H(\om)$:
\beq
\label{eq:covariance}
T_{l}^*H(\om) T_{l}=H(\om_l),
\eeq
where
\beqs
\om_l=\om +l \phi^{-1} \textnormal{mod }1.
\eeqs
Hence, for all $1\leq j<j'$, functional calculus implies
\beq
\label{eq:reduction}
\begin{aligned}
 \sum_{k=j'}^{\it} \l|\scp{\de_j}{R(E+i\eps,\om)\de_k}\r|
 &= \sum_{k=j'}^{\it}\l|\scp{\de_1}{T_{j-1}^*R(E+i\eps,\om)T_{j-1} \de_{1+k-j}} \r|\\
 &\leq \sum_{k=0}^{\it} \sup_{\om\in [0,1)} \l|\scp{\de_1}{R(E+i\eps,\om)\de_{k+N}} \r|,
\end{aligned}
\eeq
where we wrote $R(z,\om)=(-2H(\om)-z)^{-1}$ for emphasis and introduced the integer
\beqs
    N\equiv 1+j'-j
\eeqs
for notational convenience.

\subsection{Bounding Resolvent Matrix Elements by Transfer Matrix Norms}
The next step is to bound  matrix elements of the resolvent in terms of \emph{transfer matrix norms}, following the strategy of Theorem 7 in \cite{DT07} of comparing with solutions $u$ to the equation $Hu=zu$.

\be{lm}[\cite{DT07}]
\label{lm:thm7}
There exists a constant $C_2>0$ such that for  all $E\in [-K,K]$ and all $N\geq 3$, we have
\beq
\label{eq:thm7}
\sum_{k=0}^{\it} \sup_{\om\in [0,1)}\l|\scp{\de_1}{R(E+i\veps,\om)\de_{N+k}}\r| \
\leq C_2 \frac{1}{\veps (1-e^{-C_T\veps})} \min_{3\leq N_1\leq N}\sup_{\om\in [0,1)} \|\Phi_{N_1}(E+i\veps,\om)\|^{-1}.
\eeq
Here, $C_T$ is the Combes-Thomas constant and $\Phi_m(z,\om)$ is the usual transfer matrix uniquely defined by the requirement that
\beqs
 \tvector{u(m+1)}{u(m)} =\Phi_m(z,\om) \tvector{u(1)}{u(0)}, \quad \forall m\geq 0
\eeqs
holds for every half-infinite vector $u$ satisfying $H(\om)u=zu$.
\e{lm}

\be{rmk}
\be{enumerate}[label=(\roman*)]
    \item In \cite{DT07}, the left-hand side of \eqref{eq:thm7} featured an appropriate $\ell^2$-norm instead.
    \item As is emphasized in \cite{DT07}, the proof there only requires that $V$ is real and bounded. Thus, uniformity in $\om$ comes for free.
    \item The restriction that $N\geq 3$ is un-problematic for eventually concluding the LR bound, because order one terms can always be absorbed in the constant appearing on its right-hand side. As we will see later, the same idea will allow us to dispose of the potentially large $\veps$-dependent pre-factor introduced by the Combes-Thomas estimates, at the expense of increasing the exponential growth in $t$ on the right-hand side of the LR bound.
\e{enumerate}
\e{rmk}

Before the proof, we introduce some notation, which is similar to that in \cite{DT07}. For $ N\geq 3$ and $l\geq 1$, let
\beqs
    V^{\pm,N}_m(\om) =
    \be{cases}
        V_m(\om),\quad &\text{if } m\leq N\\
        \pm 2K,\quad &\text{if } m\geq N+1
    \e{cases}
\eeqs
and let $H^{\pm}_N(\om)$ be the half-infinite matrix obtained from $H(\om)$ by replacing $V_m(\om)$ with $V_m^{\pm,N}(\om)$ for all $m\geq 1$.
 Moreover, denote $R^{\pm}_N(z)=(-2H^{\pm}_N-z)^{-1}$. The crucial step is to establish the following analogue to Lemma 1 in \cite{DT07}.

\be{lm}[\cite{DT07}]
\label{lm:lm1}
We have
\beq
\label{eq:lm1}
\begin{aligned}
    &\sum_{k=0}^{\it} \sup_{\om\in [0,1)}   \l|\scp{\de_{1}}{R(E+i\veps,\om)\de_{N+k}}\r| \\
    &\leq  C \frac{1}{\veps(1-e^{-C_T\veps})}\sum_{k=0}^{\it} \sup_{\om\in [0,1)}\l|\scp{\de_{1}}{R^{\pm}_N(E+i\veps,\om)\de_{N+k}}\r| .
\end{aligned}
\eeq
\e{lm}

\be{proof}
We use the same combination of resolvent identity, resolution of the identity and Combes-Thomas estimate as in the extension from $H_n$ to $H$ before. Let $z=E+i\veps$.
By the resolvent identity,
\beqs
R(z,\om)=R^{\pm}_N(z,\om)
+ R^{\pm}_N(z,\om) \chi_N (V(\om)\mp 2K) R(z,\om),
\eeqs
where $\chi_N$ is the indicator function of the set $\set{Z}_+\setminus \{1,\ldots ,N\}$ and
\beqs
V(\om)=\sum_{m\geq 1} V_m(\om)\ketbra{\de_m}{\de_m}.
\eeqs
We use this on the left-hand side of \eqref{eq:lm1} and introduce a resolution of the identity to get
\beq
\label{eq:pff2}
\begin{aligned}
    &\sum_{k=0}^{\it}\sup_{\om\in [0,1)}    \l|\scp{\de_1}{R(z,\om)\de_{N+k}}\r| \\
    &\leq C \sum_{k=0}^{\it}\sup_{\om\in [0,1)}   \l(\l|\scp{\de_{1}}{R^{\pm}_N(z,\om)\de_{N+k}}\r|
    \vphantom{\sum_{l=0}^\it}\r.\\
    &\qquad\qquad\quad
    \l.\vphantom{\int_{-K}^K}+
    \sum_{l=0}^\it \l|\scp{\de_{1}}{R^{\pm}_N(z,\om)\de_{N+l}}
    \scp{\de_{N+l}}{R(z)\de_{N+k}}\r|\r),
\end{aligned}
\eeq
where we also used that $V(\om)\mp 2K$ is bounded, \emph{uniformly} in $\om$. By the ordinary Combes-Thomas estimate on the half-line, see (A.11) in \cite{GKT04}, we have
\beqs
     \l|\scp{\de_{N+l}}{R(z)\de_{N+k}}\r|
    \leq \frac{C}{\veps} e^{-C_T\veps |k-l|}
\eeqs
with constants that are uniform in $E,k,l$ and $\om$. Using this on \eqref{eq:pff2} and performing a geometric series, we get
\beqs
    \sum_{k=0}^{\it}\sup_{\om\in [0,1)}    \l|\scp{\de_1}{R(z,\om)\de_{N+k}}\r|
    \leq C \frac{1}{\veps(1-e^{-C_T\veps})} \sum_{k=0}^{\it}\sup_{\om\in [0,1)}    \l|\scp{\de_1}{R^{\pm}_N(z,\om)\de_{N+k}}\r|
\eeqs
and we are done.
\e{proof}

\be{proof}[Proof of Lemma \ref{lm:thm7}]
 Note that $K$ defined by \eqref{eq:Kdefn} satisfies $K\geq 4$. We follow word-for-word the proofs of Lemmas 2 and 3 in \cite{DT07} and use monotonicity of $\sqrt{\cdot}$ where appropriate. We stress that
  \be{enumerate}
  \item[(a)] these arguments only assume that $V$ is real and bounded and so none of the constants that appear depend on $\om$;
  \item[(b)] while the proof of Lemma 2 may appear to be for the whole-line case, due to the $m_-(z)$-term in formulae (35) and (36), it also applies to the half-line case with Dirichlet boundary condition (in which case the $m_-(z)$ disappears). In fact, \cite{DT07} note in the introduction that they consider both, the full-line and the half-line case, simultaneously. See also \cite{KKL}, where this method was originally developed.
  \e{enumerate}
 Together with Lemma \ref{lm:lm1}, these arguments imply that for all $N\geq 3$ and all $E\in [-K,K]$, we have
 \beqs
 \sum_{k=0}^{\it} \sup_{\om}\l|\scp{\de_1}{\frac{1}{(R(z,\om))}\de_{N+k}}\r|
 \leq C\frac{1}{\veps(1-e^{-C'\veps})}  \sup_\om\|\Phi_{N}(z,\om)\|^{-1},
 \eeqs
 for all $N\geq 3$. Since the left-hand side is monotone decreasing in $N$, we can take the minimum over $3\leq N_1\leq N$ and we are done.
 \e{proof}

\subsection{Lower Bounds on Transfer Matrix Norms and Conclusion}
Lower bounds on transfer matrix norms $\|\Phi_N(z,\om)\|$ for the Fibonacci Hamiltonian can be obtained by studying the ``trace map'', a second-order difference equation for the sequence
\beq
\label{eq:xMdefn}
x_M\equiv x_M(z,0)\equiv \text{tr } \frac{1}{2}\Phi_{F_M}(z,0),
\eeq
 where $F_M$ denotes the $M$-th Fibonacci number and the transfer matrix $\Phi$ was defined in Lemma \ref{lm:thm7}. Bounds on $|x_M|$ lead to bounds on $\|\Phi_{F_M}(z,0)\|$ via the trivial estimate
 \beq
 \label{eq:trivial}
 \|A\|\geq \frac{1}{2}|\text{tr }A|,
 \eeq
  which holds for any $2\times 2$ matrix $A$. For a detailed exposition of the trace map, we refer to Section~4 of \cite{DT07}. Here, we use an improved version of the results of \cite{DT07} established by \cite{DGY} to prove their Proposition~3.8.

\be{prop}[\cite{DGY}]
\label{prop:trace}
Let $\om=0$. There exists $\de>0$ such that for all $\veps\in(0,1]$ and all $E\in[-K,K]$, we have
\beqs
|x_M(E+i\veps,0)| \geq (1+\de)^{F_{M-M_0}}
\eeqs
for all $M\geq M_0+1$, where $M_0\geq K$ is chosen such that
\beq
\label{eq:M0defn}
 C_\de' F_{M_0}^{-s'}<\veps,
\eeq
where $s'>\al'$ and $C_\de'>0$ is an appropriate constant.
\e{prop}

\be{proof}
See the proof of Proposition 3.8 (b) in \cite{DGY}.
\e{proof}

We are now ready to give the

\be{proof}[Proof of Theorem \ref{thm:LRc}]
Recall that $N\equiv 1+j'-j$ and recall the definition of $\veps$ in \eqref{eq:epsdefn}. We start from \eqref{eq:pf2}, apply \eqref{eq:pf5} to extend to the half-line, \eqref{eq:reduction} to reduce to the case $j=1$ at the price of a $\sup_\om$ and finally we apply Lemma \ref{lm:thm7} to conclude
\beq
\label{eq:proof1}
\begin{aligned}
 \|[\tau_t^n(c_j),B]\|
 \leq &C\|B\| e^{-C_T(j'-j)}+C'\|B\| \l(\frac{1}{\veps(1-e^{-C_T\veps})}\r)^2\\
    &\times\int_{-K}^K
    \min_{3\leq N_1\leq N}\sup_{\om\in [0,1)} \|\Phi_{N_1}(E+i\veps,\om)\|^{-1}\d E.
\end{aligned}
\eeq
As remarked before, we can safely ignore the first term, since it gives an LR bound with $v=0$. Moreover, since $\|[\tau_t^n(c_j),B]\|\leq 2 \|B\|$  and the claimed LR bound allows for a constant on the right-hand side, we may ignore order-one quantities in the following.

\dashuline{ Step 1:}
We would like to apply Proposition \ref{prop:trace} together with \eqref{eq:trivial} to bound the transfer matrix norms from below by an \emph{exponentially} increasing quantity. However, it is assumed in Proposition \ref{prop:trace} that $\om=0$, while we require uniformity in $\om$. In order to extend to general $\om\in [0,1)$, we use results of \cite{D05} (this possibility was already noted in passing in \cite{DT07}): According to Proposition 3.4 in \cite{D05}, we have
\beqs
    x_M(E,\om) = x_M(E,0)
\eeqs
for all $E\in\set{R}$, all $\om\in[0,1)$ and either (a) all odd $M$ or (b) all even $M$. Both sides of this equation are complex analytic in $E$, which is obvious from the usual definition of the transfer matrices of positive index, see (12) in \cite{DT07}. Thus, we can extend the relation to \beq
\label{eq:complex}
    x_M(z,\om) = x_M(z,0),
\eeq
with $z$ complex and $M$ as before.

 \dashuline{ Step 2:}
We choose $M_1'$ to be the largest integer such that $F_{M_1'}\leq N$, i.e.\ we have
 \beqs
    F_{M_1'}\leq N < F_{M_1'+1}.
 \eeqs
If it so happens that \eqref{eq:complex} holds for all even (odd) integers $M$, but $M_1'$ is odd (even), we set $M_1\equiv M_1'-1$. Otherwise, we set $M_1\equiv M_1'$. We can assume $M_1\geq 3$, because $M_1<3$ yields an order-one bound on $N=1+j'-j$ and such terms can be ignored as we explained before.

Then, we estimate the minimum in \eqref{eq:proof1} by the $M_1$-th term and we use \eqref{eq:trivial}, \eqref{eq:complex} and Proposition \ref{prop:trace} to find, for some $\de>0$,
\beq
\label{eq:expression}
  \|[\tau_t^n(c_j),B]\|
    \leq C\|B\| e^{-C_T(j'-j)}
        + C'\|B\|\l(\frac{1}{\veps(1-e^{-C_T\veps})}\r)^2 (1+\de)^{-F_{M_1-M_0}}
\eeq
\emph{if} we have $M_1\geq M_0+1$ with $M_0$ chosen minimally, i.e.\
\beqs
C_\de' F_{M_0}^{-s'}<\veps\leq C_\de' F_{M_0-1}^{-s'}.
\eeqs
Next, we will investigate this condition further.

\dashuline{ Step 3:} The first inequality right above is equivalent to
\beq
\label{eq:Festimate}
   F_{M_0}> \l(\frac{C_\de'}{\veps}\r)^{1/s'}.
\eeq
We will use the well-known fact that
 \beq
 \label{eq:fibo}
    \frac{\phi^l}{\sqrt{5}}-\frac{1}{2}\leq F_l \leq \frac{\phi^l}{\sqrt{5}}+\frac{1}{2}
 \eeq
 for all $l\geq 0$. It implies that, up to order-one constants, one can replace $F_l$ by $\phi^l/\sqrt{5}$.
Using this and convexity of the exponential function, we conclude that the second term in \eqref{eq:expression} is bounded by
\beq
\label{eq:expression2}
C \l(\frac{1}{\veps(1-e^{-C_T\veps})}\r)^2 e^{-\mu'(F_{M_1}-F_{M_0})},
\eeq
where we introduced the positive quantity
\beqs
    \mu' =\frac{2\log(1+\de)}{\sqrt{5}}.
\eeqs
We use \eqref{eq:fibo} and recall the definitions of $M_0,M_1'$ as certain minimal/maximal integers to get
\beqs
\begin{aligned}
    F_{M_0}&\leq \phi F_{M_0-1}+2\leq \phi  \l(\frac{C_\de'}{\veps}\r)^{1/s'}+2,\\
    F_{M_1}&\geq \phi^{-1} F_{M_1'} -1 \geq \phi^{-2}(F_{M_1'+1}-C) \geq  \phi^{-2} N -C.
\end{aligned}
\eeqs
We use these to bound \eqref{eq:expression2} by
\beq
\label{eq:expression3}
C\l(\frac{1}{\veps(1-e^{-C_T\veps})}\r)^2 e^{-\mu(N-v\veps^{-1/s'})},
\eeq
with
\beq
\label{eq:vdefn}
    \mu=\phi^{-2} \mu',\qquad v=\phi^3 (C_\de')^{1/s'}
\eeq
and this bound holds for $N\geq v\veps^{-1/s'}+C'$ for some universal constant $C'$.

\dashuline{ Step 4:}
We come to the conclusion, which mainly involves making order-one changes to accommodate some exceptional cases. Suppose that $t\geq 1$, i.e.\ $\veps =t^{-1}$ according to the definition of $\veps$ in \eqref{eq:epsdefn}.
It is crucial that the pre-factor in \eqref{eq:expression3}, which quantifies the cost of our two Combes-Thomas estimates, can be bounded via
\beqs
 \l(\frac{t}{1-e^{-C_T t^{-1}}}\r)^2 \leq C_T' t^4
\eeqs
for all $t\geq 1$, where $C_T'$ is a universal constant. Moreover, $C_T' t^4$ can be bounded in terms of the exponential increase in $t$ in \eqref{eq:expression3}, by an order-one constant in front \emph{and} a change of $1/s'>\al'$ to a slightly larger value, but since $1/s'$ may be arbitrarily close to $\al'$ this change is irrelevant. We have shown that
\beq
\label{eq:proof3}
    \|[\tau_t^n(c_j),B]\| \leq C \|B\|\exp\l(  -\mu(|j'-j|-v  t^{1/s'})\r)
\eeq
for all $t\geq 1$, all $1/s'>\al'$, whenever $|j'-j|-v  t^{\al(\lam)}\geq C'$. In the case $t<1$, we have $\veps=1$ according to \eqref{eq:epsdefn} and all occurrences of $t$ in the previous argument can be replaced by order-one quantities. Since the exponential is bounded in $t<1$, we can then re-instate the $t$-dependence by yet another order-one change of the constants. Hence, \eqref{eq:proof3} extends to all $t>0$.

Finally, when $|j'-j|-v  t^{1/s'}< C'$, the argument of the exponential in the LR bound is bounded from below and hence the entire right-hand side is at least order-one. This finishes the proof.
\e{proof}

\section{Proof that $\al'=\al_u^+$}
\label{sect:al'}
In this section, we prove Proposition \ref{prop:al'}, which we recall states that
\beqs
    \al'(\lam)=\al_u^+(\lam)
\eeqs
for all $\lam>0$ (we will suppress $\lam$ from the notation from now on). The proof will proceed via the following two lemmas. The first one features some other transport exponents and we recall their definitions.

\be{defn}
We write
\beqs
    |X|^p(t)=\sum_{n>0} |n|^p |\scp{e^{-itH}\de_1}{\de_n}|^2
\eeqs
for the $p$-th moment of the position operator. For any function $f(t)$, define its \emph{time-average} by
\beqs
 \langle f\rangle(T) = \frac{2}{T} \int_0^\it e^{-2t/T}f(t) \, \d t
\eeqs
for all $T>0$. We define the transport exponents
\beq
\label{eq:betadefn}
\begin{aligned}
    \beta^+(p) & =\limsup_{t\rightarrow\it}\frac{\log|X|^p(t)}{p\log t}, \\
    \tilde{\beta}^+(p) & =\limsup_{t\rightarrow\it}\frac{\log\langle|X|^p(t)\rangle}{p\log t}.
\end{aligned}
\eeq
The time-averaged upper transport exponent can be defined by analogy with \eqref{eq:PNt} and \eqref{eq:alphauplus}, or by
\beq
\label{eq:tildealdefn}
    \tilde \al_u^+= \lim_{p\rightarrow\it}\tilde\beta^+(p).
\eeq
The two definitions are equivalent; see Theorems 2.18, 2.22 of \cite{DT10}, where it is also shown that
\beq
\label{eq:aldefn2}
    \al_u^+= \lim_{p\rightarrow\it}\beta^+(p).
\eeq
\e{defn}

\begin{lm}
\label{lm:averaging}
We have $\tilde \alpha_u^+ \le \alpha_u^+$.
\end{lm}

\begin{proof}
Fix $p \in (0,\infty)$ and take an arbitrary $\gamma > \beta^+(p)$. Then
\[
\lvert X \rvert^p(t) \le C t^{p\gamma}
\]
for some $C$ independent of $t$, so the time-averaged $p$-th moment obeys
\[
\langle \lvert X \rvert^p \rangle (T) = \frac 2T \int_0^\infty  e^{-2t/T} \lvert X \rvert^p(t) dt \le 2 C C_1 T^{p\gamma},
\]
where $C_1 = \int_0^\infty e^{-2x} x^{p\gamma} dx$. This in turn implies $\tilde\beta^+(p) \le \gamma$, which implies
\[
\tilde \beta^+(p) \le \beta^+(p)
\] for any $p\in (0,\infty)$. The claim now follows from \eqref{eq:tildealdefn} and \eqref{eq:aldefn2}.
\end{proof}

\be{lm}
\label{lm:al'}
We have $\al_u^+\leq \al'$.
\e{lm}

\be{proof}
This is just a minor modification of arguments in \cite{DGY}, in which one replaces Parseval's identity with Dunford functional calculus formula \eqref{eq:dunford} to obtain the analogue of formula (26) in \cite{DGY} with $\tilde{\al}_u^+$ replaced by $\al_u^+$. More precisely, instead of invoking \cite{DT07} to bound averaged probabilities, as done on p. 27 of \cite{DGY}, one uses the same formula without averaging derived in Theorem 1 of \cite{DT08}.

Then, part (b) of Proposition~3.8 in \cite{DGY} directly yields the same bound on $\al_u^+$ as the one on $\tilde{\al}_u^+$ in part (c) and the existence of the limit is established by \cite[Proposition 3.7]{DGY}.
\e{proof}

\be{rmk} The reason why the averaging in \cite{DGY} can be removed is that one is dealing with \emph{upper bounds} and thus the triangle inequality for integrals is available to control the possible oscillation in the Dunford functional calculus formula \eqref{eq:dunford}.
\e{rmk}

\be{proof}[Proof of Proposition \ref{prop:al'} and Corollary \ref{cor:ofproof}]
From Lemmas \ref{lm:averaging} and \ref{lm:al'}, we have, for all $\lam>0$,
\beqs
    \tilde{\al}_u^+\leq \al_u^+\leq \al'.
\eeqs
By Proposition 3.8 (c) and Proposition 3.7 in \cite{DGY}, we also have
\beqs
    \al'\leq \tilde{\al}_u^+
\eeqs
and we are done.
\e{proof}

\section{Remark on the Random Dimer Model}

We conclude with a brief discussion as to why the method of this paper will \emph{not} yield anomalous LR bounds of power-law type for the XY chain with ``random dimer'' external magnetic field. This section is mostly intended for experts and we refer to \cite{JS, JSS} for details, in particular for the precise definition of the random dimer model. The \textbf{main message} is as follows: Consider the sum over fermionic commutators in \eqref{eq:ksum}, which comes from the non-locality of the Jordan-Wigner transformation. In the Fibonacci case, the summands were decaying \emph{exponentially}, so the sum decays also exponentially and we could conclude that $LR_{\textnormal{fermi}}(\al)$ implies $LR(\al)$, with the same $\al$! On the other hand, if the fermionic commutators only decay like a \emph{power law}, as we will see is the case for the random dimer model, the sum in \eqref{eq:ksum} decreases the power-law decay by one and so, as far as our bounds go, the many-body transport is truly faster than the one-body transport on the power-law scale.

The random dimer model, introduced by Dunlap, Wu and Philips \cite{DWP}, is given by a one-dimensional discrete Laplacian together with a random potential which may take only two values $\pm \lam$ with $\lam<1$, but these values always appear in \emph{pairs}. A characteristic feature of this model is that the dimer-to-dimer transfer matrices commute at the so-called ``critical energies'' $E_c=\pm \lam$ and that in this (non-generic) case the system exhibits non-trivial transport, in contrast to the usual Anderson localization of a one-dimensional disordered quantum system.

More precisely, it follows from \cite{DT08, JS, JSS} that for the random dimer model, the transport exponent $\beta^{+}(p)$ defined in \eqref{eq:betadefn} satisfies
\beq
\label{eq:betaplusRD}
 \beta^+(p)=\max\l\{0,1-\frac{1}{2p}\r\},\quad \forall p>0.
\eeq
According to \eqref{eq:aldefn2}, we have
\beqs
\al_u^{+}=1.
\eeqs
We now consider an XY chain with external magnetic field given by pairs of random dimers $\pm \lam$. As pointed out in Remark~\ref{r.generallowerbound}, the argument that proved Theorem~\ref{thm:result2} generalizes to this case and so the best $LR(\al)$, in the sense of Definition \ref{defn:LRalpha}, that can hold for this model, is $LR(1)$, but $LR(1)$ holds for much more general models anyway \cite{NachtergaeleSims06}.

Roughly speaking, $\al_u^{+}=1$ means that the one-dimensional quantum particle has exponentially small probability to be observed a distance of order $t$ away from its initial location after time $t$ has passed. While the probability of observation is \emph{not} exponentially small for distances of order $t^\beta$ with $0 < \beta < 1$, it is \emph{polynomially} small for $\beta$ sufficiently close to $1$, since $\beta^+(p) < 1$. With this more refined perspective in mind, one could hope to prove an \emph{anomalous LR bound of power-law type} such as
\beq
\label{eq:powerlawLR}
 \ex \|[\tau_t^n(A),B]\|\leq C_1 \l(\frac{t^\beta}{|j-j'|}\r)^{\mu},
\eeq
for the random dimer model. Here, the objects $A,B,n,j,j',C,\mu$ are chosen as in Theorem \ref{thm:result}, now of course for the random dimer model, $\ex$ denotes the expectation over the randomness and $\beta>0$ should be related to $\beta^+(p)$ in some way. Let us now argue why the Jordan-Wigner method will \emph{not} give such a bound with $\beta<1$.

Following our argument for the Fibonacci case, one first proves a fermionic LR bound of power-law type. Adapting the arguments of \cite{JS} to our purposes (the main challenge again being to go from $\ell^2$-norms to $\ell^1$-norms), one finds
\beq
\label{eq:randomdimer}
 \ex \|[\tau_n^t(c_j),B]\|+ \ex \|[\tau_n^t(c_j^*),B]\|
 \leq C_1 \l( \frac{t^{\beta^+(p)+1/p+\nu}}{|j-j'|} \r)^p
\eeq
for any $\nu>0$. While it is conceivable that the extra $1/p$ term in the exponent is technical and can be removed, one still has the following problem: To obtain the LR bound for the corresponding XY chain, one has to take a sum over fermionic LR bounds, see \eqref{eq:ksum}. This yields, even without the $1/p$ term,
\beqs
    C_1  t^{p\beta^+(p)+\nu'} \sum_{l=1}^j \l( \frac{1}{l+|j'-j|}   \r)^p \leq C_1'  \frac{t^{p\beta^+(p)+\nu'}}{|j'-j|^{p-1}}
\eeqs
for any $\nu'>0$ and any $p>1$ (for $p\leq 1$, the sum diverges). Recalling \eqref{eq:betaplusRD}, we see that $p\beta^+(p)=p-1/2$ for $p>1$. Hence, the right-hand side above reads
\beqs
    C_1' \l(\frac{t^{\frac{p-1/2}{p-1}+\nu''}}{|j'-j|}\r)^{p-1}
\eeqs
for any $\nu''>0$. Of course,
\beqs
    \frac{p-1/2}{p-1}>1
\eeqs
and so this does \emph{not} yield an anomalous LR bound \eqref{eq:powerlawLR} with $\beta<1$.\footnote{Note that for $\beta\geq 1$, the LR bound of power-law type \eqref{eq:randomdimer} is weaker than $LR(1)$, which always holds and yields an exponentially small error term.} In summary, we have seen that the anomalous one-body transport of the random dimer model is still too fast to ``survive'' the summation \eqref{eq:ksum} that arises from the non-locality of the Jordan-Wigner transformation and hence too fast to yield an anomalous LR bound of power-law type on the many-body level.

\bibliographystyle{amsplain}

\begin{thebibliography}{10}

\bibitem{CombesThomas}
J.M. Combes and L.~Thomas, \emph{{Asymptotic behaviour of eigenfunctions for
  multiparticle Schr\"odinger operators}}, Comm. Math. Phys. \textbf{34}
  (1973), 251--270.

\bibitem{D05}
D.~Damanik, \emph{Dynamical upper bounds for one-dimensional quasicrystals}, J.
  Math. Anal. Appl. \textbf{303} (2005), 327--341.

\bibitem{DGLQ}
D.~Damanik, A.~Gorodetski, Q.-H.~Liu, Y.-H.~Qu, \emph{Transport exponents of Sturmian Hamiltonians}, J. Funct. Anal. \textbf{269} (2015), no. 5, 1404 –- 1440.

\bibitem{DGY}
D.~Damanik, A.~Gorodetski, and W.~Yessen, \emph{{The Fibonacci Hamiltonian}},
  arXiv:1403.7823, to appear in Invent. Math.

\bibitem{DKL00}
D.~Damanik, R.~Killip, and D.~Lenz, \emph{{Uniform spectral properties of
  one-dimensional quasicrystals. III. $\alpha$-continuity}}, Comm. Math. Phys.
  \textbf{212} (2000), 191--204.

	\bibitem{DLLY-PRL}
D.~Damanik, M.~Lemm, M.~Lukic, and W.~Yessen, \emph{{New anomalous Lieb-Robinson bounds in quasiperiodic XY chains}}, Phys. Rev. Lett. \textbf{113} (2014), 127202.

\bibitem{DLY}
D.~Damanik, M.~Lukic, and W.~Yessen, \emph{{Quantum dynamics of periodic and
  limit-periodic Jacobi and block Jacobi matrices with applications to some
  quantum many body problems}}, preprint.
	
\bibitem{DT07}
D.~Damanik and S.~Tcheremchantsev, \emph{Upper bounds in quantum dynamics}, J.
  Amer. Math. Soc. \textbf{20} (2007), 799--827.

\bibitem{DT08}
\bysame, \emph{Quantum dynamics via complex analysis methods: general upper
  bounds without time-averaging and tight lower bounds for the strongly coupled
  fibonacci hamiltonian}, J. Funct. Anal \textbf{255} (2008), 2872--2887.

\bibitem{DT10}
\bysame, \emph{{A general description of quantum dynamical spreading over an
  orthonormal basis and applications to Schr\"odinger operators}}, Discrete
  Contin. Dyn. Syst. A \textbf{28} (2010), 1381--1412.

\bibitem{DunfordSchwartz}
N.~Dunford and J.~Schwartz, \emph{{Linear Operators. Part I. General Theory}},
  Wiley, 1988.

\bibitem{DWP}
D.H. Dunlap, Wu~H.-L., and P.W. Phillips, \emph{{Absence of localization in
  random-dimer model}}, Phys. Rev. Lett. \textbf{65} (1990), 88--91.

\bibitem{GKT04}
F.~Germinet, A.~Kiselev, and S.~Tcheremchantsev, \emph{{Transfer matrices and
  transport for Schr\"odinger operators}}, Ann. Henri Poincar\'e \textbf{54}
  (2004), 2872--2887.

\bibitem{HSS11}
E.~Hamza, R.~Sims, and G.~Stolz, \emph{{Dynamical localization in disordered
  quantum spin systems}}, Comm. Math. Phys. \textbf{315} (2012), 215--239.

\bibitem{Hastings04}
M.B.~Hastings, \emph{{Lieb-Schultz-Mattis in Higher Dimensions}}, Phys. Rev. B \textbf{69} (2004), 104431.

\bibitem{Hastings07}
M.B.~Hastings, \emph{{An area law for one-dimensional quantum systems}}, J. Stat. Mech. (2007), P08024.

\bibitem{HastingsKoma06}
M.B.~Hastings and T.~Koma \emph{{Spectral Gap and Exponential Decay of Correlations}}, Comm. Math. Phys. \textbf{265}, 781 (2006).

\bibitem{JS}
S.~Jitomirskaya and H.~Schulz-Baldes, \emph{{Upper bounds on wavepacket
  spreading for random Jacobi matrices}}, Comm. Math. Phys. \textbf{273}
  (2007), 601--618.

\bibitem{JSS}
S.~Jitomirskaya, H.~Schulz-Baldes, and G.~Stolz, \emph{{Delocalization in
  random polymer models}}, Comm. Math. Phys. \textbf{233} (2003), 27--48.

\bibitem{KKL}
R.~Killip, A.~Kiselev, and Y.~Last, \emph{{Dynamical upper bounds on wavepacket
  spreading}}, Amer. J. Math. \textbf{125} (2003), 1165--1198.

\bibitem{KleinPerez}
A.~Klein and J.~W. Perez, \emph{Localization in the ground-state of the one-dimensional X-Y model with a random transverse field}, Comm. Math. Phys. \textbf{128} (1990), 99--108.

\bibitem{KruegerPHD}
H.\ Kr\"uger, \emph{{Positive Lyapunov Exponent for Ergodic Schr\"odinger
  Operators}}, Phd thesis, Rice University, 2010.

\bibitem{LSM61}
E.~Lieb, T.~Schultz, and D.~Mattis, \emph{Two soluble models of an
  antiferromagnetic chain}, Ann. Phys. \textbf{16} (1961), 407--466.

\bibitem{LiebRobinson72}
E.~H. Lieb and D.~W. Robinson, \emph{The finite group velocity of quantum spin
  systems}, Comm. Math. Phys. \textbf{28} (1972), 251--257.

\bibitem{Marin2010}
L.~Marin, \emph{{Dynamical bounds for Sturmian Schr\"odinger operators}}, Rev.
  Math. Phys. \textbf{22} (2010), 859--879.

\bibitem{NachtergaeleOgataSims06}
B.~Nachtergaele, Y.~Ogata and R.~Sims, \emph{{Propagation of Correlations in Quantum Lattice Systems}}, J. Stat. Phys. \textbf{124} (2006), 1-13.

\bibitem{NachtergaeleRazSchleinSims09}
B.~Nachtergaele, H.~Raz, B.~Schlein and R.~Sims, \emph{{Lieb-Robinson bounds for harmonic and anharmonic lattice systems}}, Comm. Math. Phys. \textbf{286} (2009), 1073-1098.


\bibitem{NachtergaeleSims06}
B.~Nachtergaele and R.~Sims, \emph{{Lieb-Robinson bounds and the exponential
  clustering theorem}}, Comm. Math. Phys. \textbf{265} (2006), 119--130.

\end{thebibliography}

\end{document}